\documentclass[11pt,a4paper,reqno]{amsart}
\usepackage{amsthm,amsmath,amsfonts,amssymb,amsxtra,appendix,bookmark,euscript,delarray,dsfont,mathrsfs}
\usepackage{enumerate}
\usepackage{amssymb}
\usepackage[utf8]{inputenc}
\usepackage{xcolor}
\usepackage{graphicx}
\usepackage{upgreek}
\usepackage[normalem]{ulem}
\usepackage{ stmaryrd }

\allowdisplaybreaks

\usepackage{tikz}
\usetikzlibrary{decorations.pathreplacing, calc}

\theoremstyle{plain}
\newtheorem{theorem}{Theorem}
\newtheorem{lemma}{Lemma}
\newtheorem*{lemma*}{Lemma}
\newtheorem{corollary}{Corollary}
\newtheorem{proposition}{Proposition}[section]

\newtheorem*{conjecture*}{Conjecture}
\newtheorem{assumption}{Assumption}
\theoremstyle{definition}

\theoremstyle{remark}
\newtheorem{remark}{Remark}

\numberwithin{equation}{section}

\renewcommand{\d}{\mathrm{d}}                       

\newcommand{\dd}{\,\mathrm{d}}    
\newcommand{\ii}{\mathrm{i}}
\newcommand{\comm}[2]{\left[#1,#2\right]}

\DeclareMathOperator{\Ker}{Ker}

\DeclareMathOperator{\tr}{Tr}

\renewcommand{\leq}{\leqslant}
\renewcommand{\geq}{\geqslant}
\renewcommand{\ge}{\geqslant}
\renewcommand{\le}{\leqslant}
\def\nn{\nonumber}

\renewcommand{\to}{\rightarrow}
\renewcommand{\phi}{\varphi}

\newcommand{\norm}[1]{\left\lVert #1 \right\rVert}
\newcommand{\abs}[1]{\left\lvert#1\right\rvert}
\renewcommand{\phi}{\varphi}
\newcommand{\ket}[1]{|#1 \rangle}       
\newcommand{\bra}[1]{\langle #1 |}      
\newcommand		{\lt}			{\left}				
\newcommand		{\rt}			{\right}

\newcommand		{\bangle}[1]	{\lt\langle #1\rt\rangle}
	
\newcommand		{\inprod}[2]	{\bangle{#1, #2}}

\newcommand\R{{\ensuremath {\mathbb R} }}

\newcommand\1{{\ensuremath {\mathds 1} }}
\newcommand{\gH}{\mathfrak{H}}

\newcommand{\gX}{\mathfrak{X}}

\newcommand{\gS}{\mathfrak{S}}

\newcommand{\cD}{\mathcal{D}}
\newcommand{\cU}{\mathcal{U}}

\newcommand{\cL}{\mathcal{L}}

\usepackage{tcolorbox}
\newtcolorbox{mybox}{colback=white,colframe=orange,boxrule=1.5pt,arc=3mm}

\title[Purified Projection Method and Uhlmann Fidelity]{Purified Projection Method and Uhlmann Fidelity for Mixed Hartree Dynamics}
\author[H. Liang]{Hao Liang}
\address{School of Mathematical Sciences, Peking University, Beijing, 100871, China}
\email{leunghao@stu.pku.edu.cn}

\author[Z. Wang]{Zhenfu Wang}
\address{Beijing International Center for Mathematical Research, Peking University, Beijing, 100871, China}
\email{zwang@bicmr.pku.edu.cn}

\begin{document}

\begin{abstract}
We give a purification and fidelity formulation of the projection
method for mixed Hartree data.  For the mean-field evolution of $N$-particle
density matrices, we prove quantitative propagation of chaos for all fixed
marginals, first in squared Uhlmann fidelity and then in trace norm via the
Fuchs--van de Graaf inequality.  The argument applies a rank-one Pickl-type
counting estimate in a purified one-particle space and uses monotonicity of
fidelity under partial trace to return to the physical variables.  The result
allows singular interactions satisfying a projected-square bound, including
$L^{2r}$ interactions and the Coulomb potential.
\end{abstract}



\maketitle

\section{Introduction}

Mean-field theory describes the effective dynamics of a large system of weakly
interacting particles by a nonlinear one-particle equation.  In the pure-state
case, one starts from the $N$-body Schr\"odinger dynamics generated by
\[
    H_N
    =
    \sum_{j=1}^N h_j
    +
    \frac1{N-1}\sum_{1\leq i<j\leq N}V(x_i-x_j),
\]
and expects initially factorized data to remain approximately factorized at
the level of fixed marginals.  More precisely, if
$\Psi_{N,0}\simeq \phi_0^{\otimes N}$, then $\Psi_{N,t}$ is expected to be
described by $\phi_t^{\otimes k}$ for every fixed $k$, where $\phi_t$ solves
the Hartree equation
\[
    \ii\partial_t\phi_t=(h+V*|\phi_t|^2)\phi_t .
\]
Equivalently, if $p_t=\ket{\phi_t}\bra{\phi_t}$, then
\[
    \ii\partial_t p_t=[h+V*\rho_{p_t},p_t],
    \qquad
    \rho_{p_t}(x)=p_t(x,x).
\]
This density-operator formulation extends naturally to mixed one-particle
states.  For a density operator $\gamma_t$, the corresponding effective
equation is the mixed Hartree equation
\[
    \ii\partial_t\gamma_t=[h+V*\rho_{\gamma_t},\gamma_t],
    \qquad
    \rho_{\gamma_t}(x)=\gamma_t(x,x).
\]
The mean-field problem is to prove that, for each fixed $k$, the physical
marginal $\Gamma_t^{N:k}$ converges to $\gamma_t^{\otimes k}$, where
$\Gamma_t^N$ solves the physical $N$-body von Neumann equation.

For pure initial data, one of the earliest results is due to
Spohn~\cite{spohn1980kinetic}, who proved trace-norm convergence of the
one-particle marginal for bounded interactions.  This was extended by Ginibre
and Velo~\cite{ginibre1979classical,ginibre1979classical2}.  Using the BBGKY
hierarchy, Erd\H{o}s and Yau~\cite{erdos2001derivation} extended Spohn's
result to the case of Coulomb interactions.  Rodnianski and
Schlein~\cite{rodnianski2009quantum} used a coherent-state approach to obtain
a quantitative trace-norm estimate for the one-particle marginal with rate
$O(N^{-1/2})$.  Knowles and Pickl~\cite{knowles2010mean} obtained a similar
rate by a projection method, allowing more singular interactions.  For
Pickl's projection method, we refer to~\cite{pickl2011simple}; more
recently, Porat and Golse~\cite{porat2024pickl} reformulated this method using
the quantum empirical measure introduced in~\cite{golse2019empirical}.  The
optimal trace-norm rate $O(N^{-1})$ was obtained by Chen, Lee, and
Schlein~\cite{chen2011rate}, including Coulomb interactions.  Related methods
also apply to many-body semi-relativistic Schr\"odinger equations describing a
boson star: Elgart and Schlein~\cite{elgart2007mean} proved convergence
without a rate, while Lee~\cite{lee2013rate} obtained the optimal rate.

For mixed initial data, Anapolitanos~\cite{anapolitanos2011rate} obtained
quantitative convergence rates for the Hartree--von Neumann mean-field limit by
generalizing Pickl's projection method \cite{pickl2011simple} to density matrices,
using Hilbert--Schmidt amplitudes of the $N$-particle and one-particle density
matrices.  Golse, Paul, and Pulvirenti~\cite{golse2018derivation} used a BBGKY
approach to obtain quantitative estimates; combined with the quantum Wasserstein
distance developed in~\cite{golse2016mean}, their method also yields bounds
uniform in the Planck constant $\hbar$.  Further results for bosonic mixed
initial data can be found in~\cite{deuchert2023dynamics}, and for fermionic
mixed initial data in~\cite{benedikter2016mean}.  More recently, in
\cite{guo2026quantum}, Guo and the authors of this paper obtained quantitative estimates for mixed initial data
by a quantum relative entropy method, motivated by the work of Jabin and the
second-named author on the classical mean-field limit~\cite{jabin2016mean,jabin2018quantitative}.

\subsection{Motivation}

The motivation for this paper comes from the formal similarity between Pickl's
projection method~\cite{pickl2011simple} and the quantum relative entropy
method~\cite{guo2026quantum}.  Both approaches rely on a Gronwall-type
argument to obtain quantitative estimates, and a key advantage of both is that
they avoid the use of BBGKY hierarchies. In the entropy
approach, differentiating the relative entropy along the physical von Neumann
flow produces a commutator involving the fluctuation Hamiltonian
$H_N-H_N^{\gamma_t}$, schematically of the form
\[
    -\ii\,
    \tr\!\left(
        \Gamma_t^N
        [H_N-H_N^{\gamma_t},\log\gamma_t^{\otimes N}]
    \right).
\]
In Pickl's method, for a pure Hartree state
$p_t=\ket{\phi_t}\bra{\phi_t}$, one introduces $q_t=\1-p_t$ and the counting
observable $\widehat n_t=N^{-1}\sum_{j=1}^Nq_j^t$.  Its derivative has the
same structural form:
\[
    \frac{\dd}{\dd t}\tr(\Gamma_t^N\widehat n_t)
    =
    \ii\,
    \tr\!\left(
        \Gamma_t^N[H_N-H_N^{p_t},\widehat n_t]
    \right),
\]
see Definition 2.2 and equation (6) in \cite{pickl2011simple}, with the choice $n(k):=k/N$, in the notation of \cite{pickl2011simple}. By replacing the counting observable $\widehat n$ with 
$\widehat n^{\,j}$, the projection method can also handle more singular
interactions of delta-function type, which is particularly useful in the
derivation of the Gross--Pitaevskii equation. See \cite{pickl2010derivation,pickl2015derivation,jeblick2019derivation}.

Thus, both methods compare the microscopic and effective dynamics through the
same fluctuation Hamiltonian.  The difference lies in the observable against
which this fluctuation is tested: the entropy method uses a logarithmic
observable, while Pickl's method uses a projection-counting observable.

The two approaches are complementary.  The projection-counting argument is
well adapted to a pure reference state.  Its basic object is the rank-one
projection $p_t$, and the passage from the counting functional to trace-norm
convergence uses this rank-one structure.  By contrast, the quantum relative
entropy method is naturally suited to faithful mixed reference states, but it
is singular for pure references due to the logarithm.  The similarity of
the two derivative identities suggests that Pickl's counting functional should
be interpreted as an information-theoretic defect which remains meaningful in
the mixed-state regime.

This defect is precisely a Uhlmann fidelity defect.  We use the squared
fidelity convention
\[
    \mathsf F(\rho,\eta):=\left(\tr\sqrt{\sqrt{\rho}\eta\sqrt{\rho}} \right)^2.
\]
If the second argument is the rank-one projection
$p_t=\ket{\phi_t}\bra{\phi_t}$, then
$\mathsf F(\rho,p_t)=\tr(\rho p_t)$. See \cite{wilde2013quantum,nielsen2010quantum} for more details. Uhlmann fidelity is closely related to quantum relative entropy through the
family of sandwiched $\alpha$-R\'enyi entropies.  In
\cite{muller2013quantum}, the authors introduced the sandwiched
$\alpha$-R\'enyi entropies as a one-parameter generalization of the quantum
relative entropy, for $\alpha\in(0,1)\cup(1,\infty)$.  At $\alpha=1/2$, this
quantity reduces, up to the conventional normalization, to the logarithm of
the Uhlmann fidelity, while in the limit $\alpha\to1$ it recovers the quantum
relative entropy.  Further properties of the sandwiched $\alpha$-R\'enyi
entropies can be found in
\cite{frank2013monotonicity,mosonyi2015quantum}. Therefore, for symmetric $\Gamma_t^N$,
\[
    \tr(\Gamma_t^N\widehat n_t)
    =
    \tr(\Gamma_t^{N:1}q_t)
    =
    1-\tr(\Gamma_t^{N:1}p_t)
    =
    1-\mathsf F(\Gamma_t^{N:1},p_t).
\]
Hence, Pickl's relative excitation number is exactly the one-particle fidelity
defect from the pure Hartree state.

The preceding identity suggests replacing $p_t$ by a mixed Hartree state
$\gamma_t$ and studying $1-\mathsf F(\Gamma_t^{N:1},\gamma_t)$.  This is a
natural one-particle quantity, but it is not a closed stability functional for
general mixed initial data.

The obstruction is dynamical.  Small one-particle fidelity defect at time
$t=0$ does not determine the correct positive-time effective trajectory.  Let
$\phi_0$ and $\psi_0$ be two orthonormal one-particle states and consider
\[
    \Gamma_0^N
    =
    \frac12\ket{\phi_0^{\otimes N}}\bra{\phi_0^{\otimes N}}
    +
    \frac12\ket{\psi_0^{\otimes N}}\bra{\psi_0^{\otimes N}}.
\]
Then
\[
    \Gamma_0^{N:1}
    =
    \gamma_0
    :=
    \frac12\ket{\phi_0}\bra{\phi_0}
    +
    \frac12\ket{\psi_0}\bra{\psi_0},
\]
and therefore $1-\mathsf F(\Gamma_0^{N:1},\gamma_0)=0$.  However, the two
branches of the microscopic state evolve independently.  At the mean-field
level this gives the mixture
\[
    \frac12\ket{\phi_t}\bra{\phi_t}
    +
    \frac12\ket{\psi_t}\bra{\psi_t},
\]
where $\phi_t$ and $\psi_t$ solve their own Hartree equations.  This is not,
in general, the solution of the mixed Hartree equation starting from
$\gamma_0$, because the vector field
$\gamma\mapsto [h+V*\rho_\gamma,\gamma]$ is nonlinear.  Thus, the condition
$1-\mathsf F(\Gamma_0^{N:1},\gamma_0)\ll1$ does not prevent the physical
evolution from selecting a different effective mixed trajectory.

This example shows that the one-particle mixed fidelity defect alone does not
encode how the mixed state is prepared.  Adding this missing information
through higher marginals would lead back to a hierarchy-type argument.  The
goal here is instead to preserve the projection-counting mechanism.

\subsection{Purification and the main idea}

The remedy is to purify the mixed Hartree state.  We enlarge the one-particle
space to $\widetilde\gH=\gH\otimes\gX$, where $\gX$ is a separable auxiliary
Hilbert space, and choose a normalized vector $\Phi_t\in\widetilde\gH$ such
that $\gamma_t=\tr_\gX\ket{\Phi_t}\bra{\Phi_t}$.  The lifted state evolves by
a Hartree equation on the enlarged space, and its partial trace solves the
mixed Hartree equation.

In the lifted space the reference state is pure again.  Hence, we can use the
genuine counting observable
\[
    \widehat n_t=\frac1N\sum_{j=1}^N q_j^t,
    \qquad
    q_j^t=\1-\ket{\Phi_t}\bra{\Phi_t}_j,
\]
and define the purified Pickl functional
\[
    \widetilde\alpha_N(t)
    =
    \left\langle
        \widetilde\Psi_{N,t},
        \widehat n_t\widetilde\Psi_{N,t}
    \right\rangle .
\]
The main estimate proves that, under a projected square bound on the
interaction singularity,
\[
    \widetilde\alpha_N(t)
    \leq
    \exp\left(C\int_0^t\Lambda_s\dd s\right)
    \left(\widetilde\alpha_N(0)+\frac1N\right).
\]
The quantity $\Lambda_t$ measures the singularity of
$V(x-y)-(V*\rho_t)(x)$ after projection onto the lifted Hartree state, where
$\rho_t(x)=\norm{\Phi_t(x)}_{\gX}^2$.  Concrete sufficient conditions for this
bound, including $L^{2r}$-type and Coulomb-type interactions, are given later.

Finally, one returns to the physical variables by the monotonicity of Uhlmann
fidelity under partial trace.  Since $\Gamma_t^{N:k}$ and
$\gamma_t^{\otimes k}$ are obtained from the lifted objects by tracing out the
auxiliary variables, the lifted counting estimate implies
\[
    1-\mathsf F(\Gamma_t^{N:k},\gamma_t^{\otimes k})
    \leq
    k\,\widetilde\alpha_N(t).
\]
The trace-norm convergence follows from the Fuchs--van de Graaf inequality in Lemma~\ref{lem:fuchs-van-de-graaf}.

The purpose of this paper is to formulate a mixed-state version of the
projection method through purification and fidelity.  A closely related
approach was developed by Anapolitanos~\cite{anapolitanos2011rate}, who used
Hilbert--Schmidt amplitudes to extend Pickl's argument to density matrices.  Here we recast the counting argument as the usual
rank-one Pickl estimate in a purified one-particle space, and then return to
the physical variables by monotonicity of Uhlmann fidelity under partial trace.

\subsection{Organization of the paper}

The paper is organized as follows.  Section~\ref{sec: purified main} states the
main fidelity propagation result and the purified Pickl estimate.  Section~\ref{sec:
preliminaries} collects the basic facts on partial traces, lifted dynamics, and
fidelity.  In Section~\ref{sec: proofs purified pickl}, we prove the main theorem
assuming the purified Pickl estimate.  Section~\ref{sec: singular interactions}
verifies the projected-square assumption for several classes of singular
interactions.  The proof of the purified Pickl estimate is given in
Appendix~\ref{sec: purified pickl}.

\subsection*{Acknowledgments} We thank Gaoyue Guo for helpful discussions. This work was partially supported by the National Key R\&D Program of China (Project No.~2024YFA1015500) and the NSFC (Grant Nos.~12595282 and 12171009).

\section{Main Results}\label{sec: purified main}

Before stating the main results, we introduce some notation. Let
\begin{equation*}
    \Omega\in\{\mathbb T^d,\R^d\},
\end{equation*}
where $\mathbb T^d$ denotes the $d$-dimensional torus and $\R^d$ denotes
the $d$-dimensional Euclidean space.  We write
\begin{equation*}
    \gH:=L^2(\Omega),
    \qquad
    \gH_N:=\gH^{\otimes N}.
\end{equation*}

Let $\cL(\gH)$ and $\cL^1(\gH)$ denote the spaces of bounded and trace-class operators on $\gH$, respectively. We denote by $\cD(\gH)$ the set of density operators on $\gH$, i.e. operators $\gamma$ satisfying 
\begin{equation*}
    \gamma=\gamma^{*}\geq 0,\qquad \tr_{\gH}(\gamma)=1.
\end{equation*}

Since we consider indistinguishable particles, we restrict attention to permutation-invariant density operators on $\gH_N$. To be precise, for each permutation $\pi\in\gS_N$, let $U_{\pi}$ be the unitary operator on $\gH_N$ defined by
\begin{equation*}
    (U_{\pi}\Psi_N)(x_1,\ldots,x_N):=\Psi_N(x_{\pi^{-1}(1)},\ldots,x_{\pi^{-1}(N)}).
\end{equation*}
We say that an $N$-particle density operator $\Gamma$ on $\gH_N$ is permutation-invariant if
\begin{equation*}
    U_{\pi}\Gamma U_{\pi}^{*}=\Gamma, \qquad \text{for all } \pi\in\gS_N.
\end{equation*}
We denote by $\cD_s(\gH_N)\subset \cD(\gH_N)$ the set of permutation-invariant $N$-particle density operators on $\gH_N$. Let $\Gamma^N\in\cD_s(\gH_N)$ and $k\leq N$. We define the $k$-th marginal $\Gamma^{N:k}$ of $\Gamma^N$ by taking the partial trace over the remaining $N-k$ particles:
\begin{equation*}
    \Gamma^{N:k}:=\tr_{k+1,\ldots,N}(\Gamma^N).
\end{equation*}

Let $h$ be a self-adjoint one-particle operator on $\gH$, and let $V$ be a real-valued even interaction potential.  We are interested in the physical $N$-body dynamics generated by
\begin{equation}\label{eq: physical HN main}
    H_N
    :=
    \sum_{j=1}^N h_j+
    \frac1{N-1}\sum_{1\leq i<j\leq N}V(x_i-x_j)
\end{equation}
on $\gH_N$.  Given an initial datum $\Gamma_0^N\in\cD_s(\gH_N)$, the evolution is described by the physical von Neumann equation
\begin{equation}\label{eq: physical von Neumann main}
    \ii\partial_t\Gamma_t^N=[H_N,\Gamma_t^N],
    \qquad
    \Gamma_t^N\big|_{t=0}=\Gamma_0^N.
\end{equation}

The expected effective dynamics for the one-particle state is the physical Hartree equation.  For a one-particle density operator $\gamma$, we define its spatial density by
\begin{equation*}
    \rho_\gamma(x):=\gamma(x,x).
\end{equation*}
Given $\gamma_0\in\cD(\gH)$, let $\gamma_t$ solve
\begin{equation}\label{eq: physical mixed Hartree main}
    \ii\partial_t\gamma_t=[h+V\ast\rho_{\gamma_t},\gamma_t],
    \qquad
    \gamma_t\big|_{t=0}=\gamma_0.
   \end{equation}
The goal is to prove a quantitative fixed marginal mean-field limit of the form
\begin{equation*}
    \Gamma_t^{N:k}\longrightarrow \gamma_t^{\otimes k},
    \qquad k\geq1 \text{ fixed}.
\end{equation*}

We use the squared Uhlmann fidelity
\begin{equation}\label{eq: Uhlmann fidelity main}
    \mathsf F(\rho,\gamma):=\norm{\sqrt\rho\sqrt\gamma}_{1}^2.
\end{equation}
If the second argument is a rank-one projection $P=\ket\phi\bra\phi$ with $\norm{\phi}=1$, then
\begin{equation*}
    \mathsf F(\rho,P)=\tr(\rho P).
\end{equation*}

From now on, we fix a constant $T>0$, and let $[0,T]$ be the time interval under consideration. We will work with the following assumptions.

\begin{assumption}\label{ass:analytic-framework}
The operator $h$ is self-adjoint and bounded from below on
$\mathfrak H=L^2(\Omega)$.  The real even interaction $V$ is form-admissible
with respect to $h$ in the sense that the mean-field $N$-body Hamiltonians defined in \eqref{eq: physical HN main}
are self-adjoint and bounded from below through closed quadratic forms.

The physical Hartree  equation \eqref{eq: physical mixed Hartree main}
is well posed on $[0,T]$ in the mild sense. 
\end{assumption}

\begin{assumption}[Projected square bound]\label{ass: projected square bound}
Define
\begin{equation}\label{eq: Lambda main}
    \Lambda_t^2
    :=
    \operatorname*{ess\,sup}_{y\in\Omega}
    \int_\Omega
    \rho_{\gamma_t}(x)\abs{V(x-y)-V\ast\rho_{\gamma_t}(x)}^2\dd x.
\end{equation}
We assume that
\begin{equation}\label{eq: Lambda integrability main}
    \int_0^T\Lambda_t\dd t<\infty
\end{equation}
for the time interval under consideration. 
\end{assumption}

\begin{remark}
Assumption~\ref{ass: projected square bound} is the singular estimate used in
the projection-counting part of the proof.
\end{remark}

\begin{theorem}[Propagation of chaos]\label{tm:fidelity-consequence}
Let $\Gamma_0^N\in\cD_s(\gH_N)$ and $\gamma_0\in\cD(\gH)$. Under Assumptions \ref{ass:analytic-framework} and \ref{ass: projected square bound}, for all $1\leq k\leq N$ and all $t\in[0,T]$,
\begin{equation}\label{eq: fidelity final bound k}
    1-
    \mathsf F\left(\Gamma_t^{N:k},\gamma_t^{\otimes k}\right)
    \leq
    2k\exp\left(8\int_0^t\Lambda_s\dd s\right)
    \left(1-\mathsf F \left( \Gamma_0^N,\gamma_0^{\otimes N} \right)  +\frac1N\right),
\end{equation}
where $\Gamma_t^N$ solves the physical von Neumann equation \eqref{eq: physical von Neumann main}, and $\gamma_t$ solves the physical Hartree equation \eqref{eq: physical mixed Hartree main}.  Consequently,
\begin{equation}\label{eq: trace final bound k}
    \norm{\Gamma_t^{N:k}-\gamma_t^{\otimes k}}_1
    \leq
    2\sqrt{2k}\,
    \exp\left(4\int_0^t\Lambda_s\dd s\right)
    \left(1-\mathsf F \left( \Gamma_0^N,\gamma_0^{\otimes N} \right)  +\frac1N\right)^{1/2}.
\end{equation}
\end{theorem}
\begin{remark}
    As explained in the Introduction, the proof of this result is based on a purification argument and a purified Pickl estimate.
\end{remark}

We next introduce a purified formulation which allows us to apply a projection counting argument in an enlarged one-particle space.

Let $\gX$ be a separable auxiliary Hilbert space.  We introduce the lifted one- and $N$-particle spaces by
\begin{equation*}
    \widetilde\gH:=\gH\otimes\gX,
    \qquad
    \widetilde\gH_N:=\widetilde\gH^{\otimes N}
    \cong
    \gH_N\otimes\gX^{\otimes N}.
\end{equation*}
When $\gH=L^2(\Omega)$, we identify $\widetilde\gH$ with $L^2(\Omega;\gX)$.  For $\Phi\in\widetilde\gH$, define its spatial density by
\begin{equation}\label{eq: density of lifted Phi}
    \rho_\Phi(x):=\norm{\Phi(x)}_{\gX}^2.
\end{equation}
If $\Phi$ is normalized, the corresponding physical one-particle density operator is
\begin{equation*}
    \gamma_\Phi:=\tr_\gX\ket{\Phi}\bra{\Phi}.
\end{equation*}
Then $\rho_\Phi(x)=\gamma_\Phi(x,x)$ in the usual density sense, as verified in the proof of Corollary~\ref{cor:lifted-Hartree-von-Neumann} below.

We shall use the following terminology for purified initial data.  Let
\[
    (\Gamma_0^N,\gamma_0)\in \cD_s(\gH_N)\times\cD(\gH).
\]
We say that the pair $(\Gamma_0^N,\gamma_0)$ admits a symmetric
$\gX$-purification $(\widetilde\Psi_{N,0},\Phi_0)$, if there exist normalized vectors
\[
    \widetilde\Psi_{N,0}\in \widetilde\gH_N,
    \qquad
    \Phi_0\in \widetilde\gH,
\]
such that
\[
    \tr_\gX \ket{\Phi_0}\bra{\Phi_0}=\gamma_0,
\qquad
    \tr_{\gX^{\otimes N}}
    \ket{\widetilde\Psi_{N,0}}\bra{\widetilde\Psi_{N,0}}
    =
    \Gamma_0^N.
\]
Moreover, $\widetilde\Psi_{N,0}$ is required to be symmetric with respect to
simultaneous permutations of the physical and auxiliary variables.  Namely, for
every $\pi\in \mathfrak{S}_N$,
\begin{equation}\label{eq:symmetric-lifted-purification}
    \widetilde U_\pi \widetilde\Psi_{N,0}
    =
    \widetilde\Psi_{N,0},
\end{equation}
where $\widetilde U_\pi$ denotes the unitary representation of $\mathfrak{S}_N$ on
$\widetilde\gH_N=(\gH\otimes\gX)^{\otimes N}$ given by
\[
    \widetilde U_\pi
    \bigl(
        (u_1\otimes\xi_1)\otimes\cdots\otimes(u_N\otimes\xi_N)
    \bigr)
    =
    (u_{\pi^{-1}(1)}\otimes\xi_{\pi^{-1}(1)})
    \otimes\cdots\otimes
    (u_{\pi^{-1}(N)}\otimes\xi_{\pi^{-1}(N)}).
\]
We lift the physical Hamiltonian to $\widetilde\gH_N$ by setting
\begin{equation}\label{eq: lifted HN main}
    \widetilde H_N
    :=
    H_N\otimes\1_{\gX^{\otimes N}}
    =
    \left(\sum_{j=1}^N h_j+
    \frac1{N-1}\sum_{1\leq i<j\leq N}V(x_i-x_j)\right)
    \otimes\1_{\gX^{\otimes N}},
\end{equation}
where both $h_j$ and $V(x_i-x_j)$ act only on the physical variables.

Let $\Phi_t\in\widetilde\gH$ be the normalized solution of the lifted Hartree equation
\begin{equation}\label{eq: lifted Hartree main}
    \ii\partial_t\Phi_t=(h+V\ast\rho_{\Phi_t})\Phi_t,
    \qquad
    \Phi_t\big|_{t=0}=\Phi_0,
\end{equation}
and let $\widetilde\Psi_{N,t}$ be the solution of the lifted Schr\"odinger equation
\begin{equation}\label{eq: lifted Schro main}
    \ii\partial_t\widetilde\Psi_{N,t}
    =
    \widetilde H_N\widetilde\Psi_{N,t},
    \qquad
    \widetilde\Psi_{N,t}\big|_{t=0}=\widetilde\Psi_{N,0}.
\end{equation}

For the lifted Hartree state $\Phi_t$, define
\begin{equation}\label{eq: p q main}
    p_j^t:=\ket{\Phi_t}\bra{\Phi_t}_j,
    \qquad
    q_j^t:=\1-p_j^t,
\end{equation}
and use the weight $n(k)=k/N$, namely
\begin{equation}\label{eq: n hat main}
    \widehat n_t:=\frac1N\sum_{j=1}^Nq_j^t.
\end{equation}
The purified Pickl functional is
\begin{equation}\label{eq: alpha tilde main}
    \widetilde\alpha_N(t)
    :=
    \left\langle\widetilde\Psi_{N,t},\widehat n_t\widetilde\Psi_{N,t}\right\rangle.
\end{equation}

\begin{theorem}[Purified Pickl estimate]\label{thm:purified-pickl}
Let $\gX$ be a separable auxiliary Hilbert space, we assume that $(\Gamma_0^N,\gamma_0)$ admits a symmetric
$\gX$-purification $(\widetilde\Psi_{N,0},\Phi_0)$. Under Assumptions \ref{ass:analytic-framework} and \ref{ass: projected square bound}, we have for all $N\geq2$ and all $t\in[0,T]$,
    \begin{equation}\label{eq: alpha gronwall bound main}
    \widetilde\alpha_N(t)
    \leq
    \exp\left(8\int_0^t\Lambda_s\dd s\right)
    \left(\widetilde\alpha_N(0)+\frac1N\right).
\end{equation}
Consequently, if $\widetilde\alpha_N(0)\to0$, then
\begin{equation*}
    \sup_{t\in[0,T]}\widetilde\alpha_N(t)\to0.
\end{equation*}
\end{theorem}

\begin{remark}
If $\Gamma_0^N=\gamma_0^{\otimes N}$ and $\Phi_0$ is any purification of $\gamma_0$, then choosing $\widetilde\Psi_{N,0}=\Phi_0^{\otimes N}$ gives $\widetilde\alpha_N(0)=0$. The purification mechanism is independent of the choice of counting weight, and
should therefore be compatible with the modified number operators used in
Pickl-type derivations of the Gross--Pitaevskii equation.  The additional
short-scale correlation estimates required in that scaling are separate
analytic issues and are not addressed here.
\end{remark}

\section{Purification and Fidelity Preliminaries}\label{sec: preliminaries}

We first record the elementary covariance property of the partial trace under a unitary acting only on the physical factor.

\begin{lemma}[Partial trace under a trivial auxiliary lift]\label{lem:partial-trace-lift}
Let $\mathcal Y$ and $\gX$ be separable Hilbert spaces and set
\[
    \widetilde{\mathcal Y}:=\mathcal Y\otimes\gX.
\]
Let $A(t)$ be a possibly time-dependent self-adjoint operator on $\mathcal Y$
which generates a unitary propagator $\cU(t,s)$ on $\mathcal Y$.  Define the
lifted propagator on $\widetilde{\mathcal Y}$ by
\[
    \widetilde \cU(t,s):=\cU(t,s)\otimes \1_\gX .
\]
Let $\widetilde\Gamma_0\in\cL^{1}(\widetilde{\mathcal Y})$ be trace class and set
\[
    \widetilde\Gamma_t
    :=
    \widetilde\cU(t,0)\widetilde\Gamma_0\widetilde\cU(t,0)^*,
    \qquad
    \Gamma_t
    :=
    \tr_\gX \widetilde\Gamma_t .
\]
Then
\[
    \Gamma_t
    =
    \cU(t,0)\Gamma_0\cU(t,0)^*,
    \qquad
    \Gamma_0:=\tr_\gX\widetilde\Gamma_0 .
\]
In particular, $\Gamma_t$ is the mild solution of the physical von Neumann
equation
\[
    \ii\partial_t\Gamma_t=\comm{A(t)}{\Gamma_t}.
\]
If, in addition, the maps are differentiable in trace norm and the commutators
are trace class, then the same identity holds in the differential sense.

The same conclusion applies to pure lifted states.  Namely, if
\[
    \widetilde\Psi_t
    =
    \widetilde\cU(t,0)\widetilde\Psi_0,
    \qquad
    \widetilde\Gamma_t
    =
    \ket{\widetilde\Psi_t}\bra{\widetilde\Psi_t},
\]
then
\[
    \Gamma_t
    :=
    \tr_\gX\ket{\widetilde\Psi_t}\bra{\widetilde\Psi_t}
\]
satisfies
\[
    \Gamma_t
    =
    \cU(t,0)
    \left(
        \tr_\gX\ket{\widetilde\Psi_0}\bra{\widetilde\Psi_0}
    \right)
    \cU(t,0)^* .
\]
\end{lemma}

\begin{proof}
The only point to prove is that partial trace is covariant under a unitary
acting on the physical factor.  We first show that for every trace class
operator $T\in\cL^1(\mathcal Y\otimes\gX)$ and every unitary $U$ on $\mathcal Y$,
\begin{equation}\label{eq:partial-trace-covariance}
    \tr_\gX\left((U\otimes\1_\gX)T(U^*\otimes\1_\gX)\right)
    =
    U(\tr_\gX T)U^* .
\end{equation}
Let $B\in\cL(\mathcal Y)$ be arbitrary.  By the defining property of the
partial trace,
\begin{align*}
&\tr_{\mathcal Y}\left[
    B\,
    \tr_\gX\left((U\otimes\1_\gX)T(U^*\otimes\1_\gX)\right)
\right] \\
&\quad =
\tr_{\mathcal Y\otimes\gX}
\left[
    (B\otimes\1_\gX)
    (U\otimes\1_\gX)
    T
    (U^*\otimes\1_\gX)
\right].
\end{align*}
Using cyclicity of the trace, the right-hand side becomes
\begin{align*}
\tr_{\mathcal Y\otimes\gX}
\left[
    (U^*BU\otimes\1_\gX)T
\right]
=
\tr_{\mathcal Y}
\left[
    U^*BU\,\tr_\gX T
\right]
=
\tr_{\mathcal Y}
\left[
    B\,U(\tr_\gX T)U^*
\right].
\end{align*}
Since this holds for all $B\in\cL(\mathcal Y)$, \eqref{eq:partial-trace-covariance}
follows.

We now apply \eqref{eq:partial-trace-covariance} with
\[
    U=\cU(t,0),
    \qquad
    T=\widetilde\Gamma_0 .
\]
Since
\[
    \widetilde\Gamma_t
    =
    (\cU(t,0)\otimes\1_\gX)
    \widetilde\Gamma_0
    (\cU(t,0)^*\otimes\1_\gX),
\]
we get
\[
    \Gamma_t
    =
    \tr_\gX\widetilde\Gamma_t
    =
    \cU(t,0)
    (\tr_\gX\widetilde\Gamma_0)
    \cU(t,0)^*
    =
    \cU(t,0)\Gamma_0\cU(t,0)^* .
\]
This is precisely the mild solution of
\[
    \ii\partial_t\Gamma_t=\comm{A(t)}{\Gamma_t}.
\]
The statement for pure lifted states follows by taking
$
    \widetilde\Gamma_t
    =
    \ket{\widetilde\Psi_t}\bra{\widetilde\Psi_t}.
$
\end{proof}

\begin{corollary}[Application to the lifted Schrödinger equation]
\label{cor:lifted-von-Neumann}
Note that
\[
    \widetilde\gH_N
    \cong
    \gH_N\otimes\gX^{\otimes N}
\]
and the lifted Hamiltonian has the form
\[
    \widetilde H_N
    =
    H_N\otimes\1_{\gX^{\otimes N}},
\]
where $H_N$ acts on $\gH_N$.  Let
\[
    \widetilde\Gamma_t^N
    =
    \widetilde\cU_N(t,0)\widetilde\Gamma_0^N\widetilde\cU_N(t,0)^*
\]
be the lifted evolution, with
\[
    \widetilde\cU_N(t,s)
    =
    \cU_N(t,s)\otimes\1_{\gX^{\otimes N}} .
\]
Then $\tr_{\gX^{\otimes N}}\widetilde{\Gamma}_t^N$ solves the physical von Neumann equation \eqref{eq: physical von Neumann main} in the mild sense.
\end{corollary}

\begin{proof}
This is Lemma \ref{lem:partial-trace-lift} with
\[
    \mathcal Y=\gH_N,
    \qquad
    \gX \text{ replaced by } \gX^{\otimes N},
    \qquad
    A(t)=H_N.
\]
\end{proof}

\begin{corollary}[Application to the lifted Hartree equation]
\label{cor:lifted-Hartree-von-Neumann}
Let $\gH=L^2(\Omega)$ and $\widetilde\gH=\gH\otimes\gX$.  Let
\[
    \Phi_t\in \widetilde\gH\simeq L^2(\Omega;\gX)
\]
solve
\[
    \ii\partial_t\Phi_t
    =
    \big(h+V\ast\rho_{\Phi_t}\big)\Phi_t,
\]
where $h$ and $W_t$ act only on the physical variable. Then $\widetilde{\gamma}_t:=\tr_{\gX}\ket{\Phi_t}\bra{\Phi_t}$ solves the physical Hartree equation \eqref{eq: physical mixed Hartree main} in the mild sense.
\end{corollary}

\begin{proof}
Apply Lemma \ref{lem:partial-trace-lift} with
\[
    \mathcal Y=\gH,
    \qquad
    A(t)=h+V\ast\rho_{\Phi_t},
    \qquad
    \widetilde\Gamma_t=\ket{\Phi_t}\bra{\Phi_t}.
\]
It remains only to identify the density.  Let $(e_m)_{m\geq1}$ be an
orthonormal basis of $\gX$ and define scalar functions
\[
    \phi_m(t,x):=\left\langle e_m,\Phi_t(x)\right\rangle_{\gX}.
\]
Then the partial trace can be written as
\[
    \widetilde{\gamma}_t
    =
    \sum_{m\geq1}\ket{\phi_m(t)}\bra{\phi_m(t)}
\]
with convergence in trace norm.  Therefore, the integral kernel of $\widetilde{\gamma}_t$
is
\[
    \widetilde{\gamma}_t(x,y)
    =
    \sum_{m\geq1}\phi_m(t,x)\overline{\phi_m(t,y)}
    =
    \left\langle \Phi_t(y),\Phi_t(x)\right\rangle_{\gX}.
\]
In particular,
\[
    \widetilde{\gamma}_t(x,x)
    =
    \sum_{m\geq1}|\phi_m(t,x)|^2
    =
    \norm{\Phi_t(x)}_{\gX}^2
    =
    \rho_t(x).
\]
Hence, $W_t=V*\rho_{\widetilde{\gamma}_t}$, and the equation becomes
\[
    \ii\partial_t\widetilde{\gamma}_t
    =
    \comm{h+V*\rho_{\widetilde{\gamma}_t}}{\widetilde{\gamma}_t}.
\]
\end{proof}

The following two lemmas are standard and will be used in the proof of Theorem~\ref{tm:fidelity-consequence}. Recall that the squared Uhlmann fidelity is defined in \eqref{eq: Uhlmann fidelity main}.

\begin{lemma}[Fuchs--van de Graaf inequality]
\label{lem:fuchs-van-de-graaf}
Let $\mathcal Y$ be a separable Hilbert space and let
$\rho,\gamma\in\cD(\mathcal Y)$. Then
\begin{equation}\label{eq:fuchs-van-de-graaf-squared}
    \frac12\norm{\rho-\gamma}_1
    \leq
    \sqrt{1-\mathsf F(\rho,\gamma)}.
\end{equation}
\end{lemma}

\begin{proof}
Let $P$ be the spectral projection of $\rho-\gamma$ corresponding to the
positive spectrum.  Since $\rho-\gamma$ is self-adjoint, trace class, and has
trace zero, its positive and negative parts have the same trace.  Hence,
\begin{equation}\label{eq:positive-projection-trace-distance}
    \tr\big(P(\rho-\gamma)\big)
    =
    \frac12\norm{\rho-\gamma}_1.
\end{equation}
In particular,
\[
    \tr(P\rho)\geq \tr(P\gamma).
\]

We estimate the root fidelity directly.  Since $P+(\1-P)=\1$, the triangle
inequality gives
\begin{align}
    \norm{\sqrt\rho\sqrt\gamma}_1
    &=
    \norm{\sqrt\rho P\sqrt\gamma
    +
    \sqrt\rho(\1-P)\sqrt\gamma}_1
    \nn\\
    &\leq
    \norm{\sqrt\rho P\sqrt\gamma}_1
    +
    \norm{\sqrt\rho(\1-P)\sqrt\gamma}_1 .
\end{align}
By H\"older's inequality for Schatten norms,
\[
    \norm{\sqrt\rho P\sqrt\gamma}_1
    \leq
    \norm{\sqrt\rho P}_2\norm{P\sqrt\gamma}_2
    =
    \sqrt{\tr(P\rho)}\sqrt{\tr(P\gamma)}.
\]
Similarly,
\[
    \norm{\sqrt\rho(\1-P)\sqrt\gamma}_1
    \leq
    \sqrt{\tr((\1-P)\rho)}
    \sqrt{\tr((\1-P)\gamma)}.
\]
Therefore,
\begin{equation}\label{eq:root-fidelity-two-partition-bound}
    \norm{\sqrt\rho\sqrt\gamma}_1
    \leq
    \sqrt{\tr(P\rho)\tr(P\gamma)}
    +
    \sqrt{\tr((\1-P)\rho)\tr((\1-P)\gamma)}.
\end{equation}

We now use the elementary scalar inequality
\begin{align}
\left(
    \sqrt{\tr(P\rho)\tr(P\gamma)}
    +
    \sqrt{\tr((\1-P)\rho)\tr((\1-P)\gamma)}
\right)^2
\leq
1-\left(\tr(P(\rho-\gamma))\right)^2.
\label{eq:binary-fidelity-direct}
\end{align}
Indeed, after expanding both sides, \eqref{eq:binary-fidelity-direct} is
equivalent to
\[
    \left(
        \sqrt{\tr(P\rho)\tr((\1-P)\rho)}
        -
        \sqrt{\tr(P\gamma)\tr((\1-P)\gamma)}
    \right)^2
    \geq0.
\]
Combining \eqref{eq:root-fidelity-two-partition-bound} and
\eqref{eq:binary-fidelity-direct}, we obtain
\[
    \mathsf F(\rho,\gamma)
    =
    \norm{\sqrt\rho\sqrt\gamma}_1^2
    \leq
    1-\left(\tr(P(\rho-\gamma))\right)^2.
\]
Using \eqref{eq:positive-projection-trace-distance}, this becomes
\[
    \mathsf F(\rho,\gamma)
    \leq
    1-\frac14\norm{\rho-\gamma}_1^2.
\]
This proves \eqref{eq:fuchs-van-de-graaf-squared}.
\end{proof}

\begin{lemma}[Data processing inequality]
\label{lem:fidelity-data-processing}
Let $\mathcal Y_1$ and $\mathcal Y_2$ be separable Hilbert spaces, and let
\[
    \mathcal T:\cL^1(\mathcal Y_1)\to\cL^1(\mathcal Y_2)
\]
be a completely positive trace-preserving map.  Then, for all
$\rho,\gamma\in\cD(\mathcal Y_1)$,
\begin{equation}\label{eq:fidelity-data-processing}
    \mathsf F\left(\mathcal T(\rho),\mathcal T(\gamma)\right)
    \geq
    \mathsf F(\rho,\gamma).
\end{equation}
\end{lemma}
\begin{proof}
This is a standard consequence of the data processing inequality for the
sandwiched $\alpha$-R\'enyi relative entropy.  Indeed, applying
Theorem~1 of~\cite{frank2013monotonicity} with $\alpha=1/2$ gives precisely
the monotonicity of the squared Uhlmann fidelity under completely positive
trace-preserving maps.
\end{proof}

\begin{remark}
We use only the monotonicity of the fidelity under completely positive trace-preserving maps.  In the present argument, the relevant map is the partial trace.
\end{remark}

\section{Proofs of the Main Results}\label{sec: proofs purified pickl}

In this section we prove Theorem~\ref{tm:fidelity-consequence} assuming Theorem~\ref{thm:purified-pickl}. Since the proof of Theorem~\ref{thm:purified-pickl} follows from a modification of the original Pickl argument \cite{pickl2011simple}, its details are postponed to Appendix~\ref{sec: purified pickl} for completeness.

\begin{proof}[Proof of Theorem~\ref{tm:fidelity-consequence}]
Let $\gX$ be a separable auxiliary Hilbert space, we first prove an estimate for an arbitrary symmetric $\gX$-purification. Let $(\Gamma_0^N,\gamma_0)$ admits a symmetric
$\gX$-purification $(\widetilde\Psi_{N,0},\Phi_0)$. Let
\begin{equation*}
    \widetilde\Gamma_t^N:=\ket{\widetilde\Psi_{N,t}}\bra{\widetilde\Psi_{N,t}},
    \qquad
     P_t:=\ket{\Phi_t}\bra{\Phi_t}.
\end{equation*}
Using Corollaries~\ref{cor:lifted-von-Neumann} and
\ref{cor:lifted-Hartree-von-Neumann}, together with the uniqueness assumptions in Assumption~\ref{ass:analytic-framework}, we have for $t\in[0,T]$ that
\begin{equation*}
    \Gamma_t^N=\tr_{\gX^{\otimes N}}\widetilde\Gamma_t^N,\qquad \gamma_t=\tr_{\gX}P_t.
\end{equation*}
For $1\leq k\leq N$, denote by $\widetilde\Gamma_t^{N:k}$ and $\Gamma_t^{N:k}$ the $k$-particle marginals of $\widetilde\Gamma_t^N$ and $\Gamma_t^N$, respectively.

For the lifted $k$-particle marginal we use the notation
\begin{equation*}
    P_t^{\otimes k}=p_1^t\cdots p_k^t.
\end{equation*}
Since the projections $q_j^t=\1-p_j^t$ commute for different $j$,
\begin{equation*}
    \1-p_1^t\cdots p_k^t
    =
    \1-
    \prod_{j=1}^k(\1-q_j^t)
    \leq
    \sum_{j=1}^k q_j^t.
\end{equation*}
Therefore, by symmetry,
\begin{align*}
&1-
\tr_{\widetilde\gH^{\otimes k}}
\left(\widetilde\Gamma_t^{N:k}P_t^{\otimes k}\right) \\
&\quad =
\tr_{\widetilde\gH_N}
\left(
\widetilde\Gamma_t^N
(\1-p_1^t\cdots p_k^t)
\right) \leq
\sum_{j=1}^k
\tr_{\widetilde\gH_N}
\left(
\widetilde\Gamma_t^Nq_j^t
\right)
=
 k\,\widetilde\alpha_N(t).
\end{align*}
Since $P_t^{\otimes k}$ is a rank-one projection,
\begin{equation*}
    \mathsf F\left(\widetilde\Gamma_t^{N:k},P_t^{\otimes k}\right)
    =
    \tr_{\widetilde\gH^{\otimes k}}
    \left(\widetilde\Gamma_t^{N:k}P_t^{\otimes k}\right).
\end{equation*}
Thus,
\begin{equation*}
    1-
    \mathsf F\left(\widetilde\Gamma_t^{N:k},P_t^{\otimes k}\right)
    \leq
    k\,\widetilde\alpha_N(t).
\end{equation*}
The partial trace $\tr_{\gX^{\otimes k}}$ is a completely positive trace-preserving map. Using Lemma~\ref{lem:fidelity-data-processing}, we obtain
\begin{align*}
    \mathsf F\left(\Gamma_t^{N:k},\gamma_t^{\otimes k}\right)
    =
    \mathsf F\left(
    \tr_{\gX^{\otimes k}}\widetilde\Gamma_t^{N:k},
    \tr_{\gX^{\otimes k}}P_t^{\otimes k}
    \right) \geq
    \mathsf F\left(\widetilde\Gamma_t^{N:k},P_t^{\otimes k}\right),
\end{align*}
which, combined with Theorem~\ref{thm:purified-pickl}, gives
\begin{equation}\label{eq:fidelity-alpha-bound}
    1-\mathsf F\left(\Gamma_t^{N:k},\gamma_t^{\otimes k}\right)
    \leq
    k\,\widetilde\alpha_N(t)
    \leq
    k\exp\left(8\int_0^t\Lambda_s\dd s\right)\left(\widetilde\alpha_N(0)+\frac1N\right).
\end{equation}
Now we choose a concrete symmetric $\gX$-purification
$(\widetilde\Psi_{N,0},\Phi_0)$ of $(\Gamma_0^N,\gamma_0)$.

We take
\[
    \gX=\overline{\gH}.
\]
Then
\[
    \widetilde\gH_N
    =
    (\gH\otimes\overline{\gH})^{\otimes N}
    \cong
    \gH_N\otimes\overline{\gH_N}.
\]
We use the Hilbert--Schmidt vectorization
\[
    \gH_N\otimes\overline{\gH_N}
    \cong
    \mathcal L^2(\gH_N),
    \qquad
    K\mapsto |K\rangle\rangle.
\]
Under this identification,
\[
    \tr_{\overline{\gH_N}}
    |K\rangle\rangle\langle\langle K|
    =
    KK^*,
\]
and, if $U_\pi$ denotes the permutation representation on $\gH_N$,
\[
    \widetilde U_\pi |K\rangle\rangle
    =
    |U_\pi K U_\pi^*\rangle\rangle.
\]

Since both $\Gamma_0^N$ and $\gamma_0^{\otimes N}$ are permutation-invariant, the operator
$
    A:=\sqrt{\Gamma_0^N}\sqrt{\gamma_0^{\otimes N}}
$
commutes with all $U_\pi$.  Let
\[
    \mathcal M
    :=
    \{B\in\cL(\gH_N): U_\pi BU_\pi^*=B,\ \forall\pi\in \mathfrak{S}_N\}.
\]
Then $\mathcal M$ is a unital $C^*$-algebra and $A\in\mathcal M\cap\mathcal L^1(\gH_N)$.
Averaging over the finite group $\mathfrak{S}_N$ shows that
\[
    \|A\|_1
    =
    \sup_{\substack{B\in\mathcal M\\ \|B\|_{\rm op}\le1}}
    \left|
    \tr(B^*A)
    \right|.
\]
Since the closed unit ball of a unital $C^*$-algebra is the closed convex hull
of its unitaries, for every $\varepsilon>0$ there exists a unitary
$V_\varepsilon\in\mathcal M$ such that
\[
    \left|
    \tr(V_\varepsilon^*A)
    \right|^2
    \ge
    \|A\|_1^2-\varepsilon.
\]
We choose $\varepsilon=1/N$.  Since
$
    \|A\|_1^2
    =
    \mathsf F(\Gamma_0^N,\gamma_0^{\otimes N})
$,
we get
\[
    \left|
    \tr(V_{1/N}^*
    \sqrt{\Gamma_0^N}
    \sqrt{\gamma_0^{\otimes N}})
    \right|^2
    \ge
    \mathsf F(\Gamma_0^N,\gamma_0^{\otimes N})-\frac1N.
\]

Define
$
    K_{1/N}:=\sqrt{\Gamma_0^N}\,V_{1/N}
$. Then
$
    K_{1/N} K_{1/N}^*=\Gamma_0^N
$.
Moreover, since $\sqrt{\Gamma_0^N}$ and $V_{1/N}$ both belong to
$\mathcal M$,
\[
    U_\pi K_{1/N} U_\pi^*=K_{1/N},
    \qquad \forall\pi\in \mathfrak{S}_N.
\]
We now set
\[
    \widetilde\Psi_{N,0}:=|K_{1/N}\rangle\rangle,
    \qquad
    \Phi_0:=|\sqrt{\gamma_0}\rangle\rangle.
\]
Then
$$
    \tr_{\overline{\gH_N}}
    |\widetilde\Psi_{N,0}\rangle\langle\widetilde\Psi_{N,0}|
    =
    \Gamma_0^N,
\quad
    \widetilde U_\pi\widetilde\Psi_{N,0}
    =
    \widetilde\Psi_{N,0},
    \forall\pi\in S_N.
\quad
    \tr_{\overline{\gH}}|\Phi_0\rangle\langle\Phi_0|
    =
    \gamma_0.
$$
Thus $(\widetilde\Psi_{N,0},\Phi_0)$ is a symmetric
$\overline{\gH}$-purification of $(\Gamma_0^N,\gamma_0)$.
Furthermore,
$
    \Phi_0^{\otimes N}
    =
    |\sqrt{\gamma_0^{\otimes N}}\rangle\rangle
$.
Hence
\begin{align*}
    \left|
    \left\langle
    \widetilde\Psi_{N,0},
    \Phi_0^{\otimes N}
    \right\rangle
    \right|^2
    =
    \left|
    \tr\left(
    V_{1/N}^*
    \sqrt{\Gamma_0^N}
    \sqrt{\gamma_0^{\otimes N}}
    \right)
    \right|^2  
    \ge
    \mathsf F(\Gamma_0^N,\gamma_0^{\otimes N})-\frac1N .
\end{align*}

We now estimate the initial lifted counting functional.  Since
\[
    |\Phi_0^{\otimes N}\rangle\langle\Phi_0^{\otimes N}|
    =
    \prod_{j=1}^N p_j^0
    \le
    \frac1N\sum_{j=1}^N p_j^0,
\]
we have
\begin{align*}
    \left|
    \left\langle
    \widetilde\Psi_{N,0},
    \Phi_0^{\otimes N}
    \right\rangle
    \right|^2
    \le
    \left\langle
    \widetilde\Psi_{N,0},
    \frac1N\sum_{j=1}^N p_j^0
    \widetilde\Psi_{N,0}
    \right\rangle  
    =
    1-\widetilde\alpha_N(0).
\end{align*}
Consequently,
\[
    \widetilde\alpha_N(0)
    \le
    1-\mathsf F(\Gamma_0^N,\gamma_0^{\otimes N})
    +\frac1N.
\]
Substituting this into \eqref{eq:fidelity-alpha-bound}, we obtain
\begin{align*}
    1-\mathsf F\left(\Gamma_t^{N:k},\gamma_t^{\otimes k}\right)
    &\le
    k\exp\left(8\int_0^t\Lambda_s\dd s\right)
    \left(
    1-\mathsf F(\Gamma_0^N,\gamma_0^{\otimes N})
    +\frac2N
    \right)  \\
    &\le
    2k\exp\left(8\int_0^t\Lambda_s\dd s\right)
    \left(
    1-\mathsf F(\Gamma_0^N,\gamma_0^{\otimes N})
    +\frac1N
    \right).
\end{align*}
This proves \eqref{eq: fidelity final bound k}.

Finally, the Fuchs--van de Graaf inequality in Lemma~\ref{lem:fuchs-van-de-graaf} gives
\[
    \norm{\Gamma_t^{N:k}-\gamma_t^{\otimes k}}_1
    \le
    2
    \left(
    1-\mathsf F(\Gamma_t^{N:k},\gamma_t^{\otimes k})
    \right)^{1/2}.
\]
Using \eqref{eq: fidelity final bound k}, we get
\[
    \norm{\Gamma_t^{N:k}-\gamma_t^{\otimes k}}_1
    \le
    2\sqrt{2k}\,
    \exp\left(4\int_0^t\Lambda_s\dd s\right)
    \left(
    1-\mathsf F(\Gamma_0^N,\gamma_0^{\otimes N})
    +\frac1N
    \right)^{1/2}.
\]
This proves \eqref{eq: trace final bound k}.
\end{proof}

\section{Sufficient Conditions for Singular Interactions}\label{sec: singular interactions}

The condition \eqref{eq: Lambda main} is the intrinsic condition used in the proof.  It is useful to record two simple criteria which imply it. Recall that from Corollary~\ref{cor:lifted-Hartree-von-Neumann} and Assumption~\ref{ass:analytic-framework}, we have $\rho_{\gamma_t}=\rho_{\Phi_t}$ for the time interval under consideration.

\begin{proposition}\label{prop:Lp-condition-Lambda}
Assume that $V\in L^{2r}(\Omega)$ for some $r\geq1$ and set $s=r/(r-1)$, with the convention $s=\infty$ when $r=1$.  If $\rho_{\gamma_t}\in L^s(\Omega)$, then
\begin{equation}\label{eq: Lambda Lp condition}
    \Lambda_t
    \leq
    2\norm{V}_{L^{2r}(\Omega)}\norm{\rho_{\gamma_t}}_{L^s(\Omega)}^{1/2}.
\end{equation}
Consequently, $\int_0^T\Lambda_t\dd t<\infty$ follows from
\begin{equation*}
    \int_0^T\norm{\rho_{\gamma_t}}_{L^s(\Omega)}^{\frac{1}{2}}\dd t<\infty.
\end{equation*}
\end{proposition}

\begin{proof}
Let
\begin{equation*}
    K_t^2
    :=
    \operatorname*{ess\,sup}_{y\in\Omega}
    \int_\Omega \rho_{\gamma_t}(x)\abs{V(x-y)}^2\dd x.
\end{equation*}
By H\"older's inequality,
\begin{equation*}
    K_t^2
    \leq
    \norm{V}_{L^{2r}}^2\norm{\rho_{\gamma_t}}_{L^s}.
\end{equation*}
Moreover,
\begin{align*}
    \int \rho_{\gamma_t}(x)\abs{V\ast\rho_{\gamma_t}(x)}^2\dd x
    &\leq
    \int \rho_{\gamma_t}(x)(\abs V^2*\rho_{\gamma_t})(x)\dd x
    \leq K_t^2.
\end{align*}
Using $\abs{a-b}^2\leq 2\abs a^2+2\abs b^2$ in \eqref{eq: Lambda main}, we obtain $\Lambda_t^2\leq4K_t^2$, proving the claim.
\end{proof}

\begin{proposition}[Coulomb-type control]\label{prop:coulomb-control-Lambda}
Let $\Omega=\R^3$ and $V(x)=\lambda\abs{x}^{-1}$.  If $\Phi_t\in H^1(\R^3;\gX)$, then
\begin{equation}\label{eq: Coulomb Lambda condition}
    \Lambda_t
    \leq
    4\abs\lambda\norm{\nabla\Phi_t}_{L^2(\R^3;\gX)}.
\end{equation}
Equivalently,
\begin{equation*}
    \Lambda_t
    \leq
    4\abs\lambda\left(\tr_\gH(-\Delta\gamma_t)\right)^{1/2},
    \qquad
    \gamma_t=\tr_\gX\ket{\Phi_t}\bra{\Phi_t}.
\end{equation*}
\end{proposition}

\begin{proof}
By the vector-valued Hardy inequality,
\begin{equation*}
    \int_{\R^3}\frac{\norm{\Phi_t(x)}_{\gX}^2}{\abs{x-y}^2}\dd x
    \leq
    4\int_{\R^3}\norm{\nabla\Phi_t(x)}_{\gX}^2\dd x
\end{equation*}
for every $y\in\R^3$.  Therefore,
\begin{equation*}
    K_t^2
    :=
    \sup_y\int \rho_t(x)\abs{V(x-y)}^2\dd x
    \leq
    4\abs\lambda^2\norm{\nabla\Phi_t}_{L^2(\R^3;\gX)}^2.
\end{equation*}
As in the proof of Proposition \ref{prop:Lp-condition-Lambda}, $\Lambda_t\leq2K_t$, which gives \eqref{eq: Coulomb Lambda condition}.
\end{proof}

\appendix 

\section{The purified Pickl estimate}\label{sec: purified pickl}

\begin{proof}[Proof of Theorem \ref{thm:purified-pickl}]
We write $p_j,q_j,\widehat n$ instead of $p_j^t,q_j^t,\widehat n_t$ when no confusion is possible.  Define the lifted mean-field Hamiltonian
\begin{equation*}
    \widetilde H_N^{\Phi_t}:=\sum_{j=1}^N(h_j+V\ast\rho_{\Phi_t}(x_j)).
\end{equation*}
Since $\partial_t q_j^t=-\ii[h_j+V\ast\rho_{\Phi_t}(x_j),q_j^t]$, one has
\begin{equation*}
    \partial_t\widehat n_t=-\ii[\widetilde H_N^{\Phi_t},\widehat n_t].
\end{equation*}
Using \eqref{eq: lifted Schro main} and the convention that the scalar product is linear in the second variable, we obtain
\begin{align}\label{eq:alpha-derivative-commutator}
    \frac{\dd}{\dd t}\widetilde\alpha_N(t)
    &=\ii\left\langle\widetilde\Psi_{N,t},[\widetilde H_N,\widehat n_t]\widetilde\Psi_{N,t}\right\rangle
      -\ii\left\langle\widetilde\Psi_{N,t},[\widetilde H_N^{\Phi_t},\widehat n_t]\widetilde\Psi_{N,t}\right\rangle \nn\\
    &=\ii\left\langle\widetilde\Psi_{N,t},[\widetilde H_N-\widetilde H_N^{\Phi_t},\widehat n_t]\widetilde\Psi_{N,t}\right\rangle .
\end{align}
By symmetry of $\widetilde\Psi_{N,t}$ and the normalization $1/(N-1)$ in \eqref{eq: lifted HN main}, \eqref{eq:alpha-derivative-commutator} reduces to
\begin{equation}\label{eq:alpha-derivative-D}
    \frac{\dd}{\dd t}\widetilde\alpha_N(t)
    =\ii\left\langle\widetilde\Psi_{N,t},[D_t,q_1]\widetilde\Psi_{N,t}\right\rangle,
\end{equation}
where
\begin{equation}\label{eq: D t main}
    D_t:=V(x_1-x_2)-V\ast\rho_{\Phi_t}(x_1).
\end{equation}
This is the two-body fluctuation multiplication operator defined on $\widetilde{\gH}_N$. The following properties of $D_t$ are crucial for the proof:

\begin{lemma}
\label{lem:condensate-cancellation-projected-square}
Let $D_t=(D_t)_{12}$ be the two-body fluctuation multiplication operator defined in \eqref{eq: D t main}.
Then the following two estimates hold.

First, we have the following cancellation rule:
\begin{equation}\label{eq:p2Dp2-cancellation-lemma}
    p_{2}^tD_t p_{2}^t=0.
\end{equation}

Second, the singularity of $D_t$ is controlled after projection:
\begin{equation}\label{eq:p1D2p1-projected-square-lemma}
    \norm{p_{1}^tD_t^2p_{1}^t}_{\rm op}
    \le
    \Lambda_t^2.
\end{equation}
Equivalently,
\[
    \norm{D_t p_{1}^t}_{\rm op}\le \Lambda_t.
\]
If $D_t$ is not globally bounded, then \eqref{eq:p1D2p1-projected-square-lemma}
is understood in the quadratic form sense.
\end{lemma}

We now estimate the right-hand side of \eqref{eq:alpha-derivative-D} by using Lemma~\ref{lem:condensate-cancellation-projected-square}.  Since $D_t$ and $q_1$ are self-adjoint,
\begin{equation}\label{eq:comm-to-pDq}
    \left|\frac{\dd}{\dd t}\widetilde\alpha_N(t)\right|
    \leq
    2\abs{\left\langle\widetilde\Psi_{N,t},p_1D_tq_1\widetilde\Psi_{N,t}\right\rangle}.
\end{equation}
Insert $\1=p_2+q_2$ on both sides of $D_t$.  The term with $p_2$ on both sides vanishes by \eqref{eq:p2Dp2-cancellation-lemma} in Lemma~\ref{lem:condensate-cancellation-projected-square}:
\begin{equation*}
    \left\langle\widetilde\Psi,p_1p_2D_tq_1p_2\widetilde\Psi\right\rangle=0.
\end{equation*}
We are left with
\begin{align*}
    T_{01}:=\left\langle\widetilde\Psi,p_1p_2D_tq_1q_2\widetilde\Psi\right\rangle,\ 
    T_{10}:=\left\langle\widetilde\Psi,p_1q_2D_tq_1p_2\widetilde\Psi\right\rangle,\ 
    T_{11}:=\left\langle\widetilde\Psi,p_1q_2D_tq_1q_2\widetilde\Psi\right\rangle,
\end{align*}
where $\widetilde\Psi=\widetilde\Psi_{N,t}$.

For $T_{10}$, using \eqref{eq:p1D2p1-projected-square-lemma} in Lemma~\ref{lem:condensate-cancellation-projected-square} and symmetry,
\begin{align*}
    |T_{10}|
    =\abs{\left\langle D_t p_1q_2\widetilde\Psi,q_1p_2\widetilde\Psi\right\rangle}
    \leq \Lambda_t\norm{q_2\widetilde\Psi}\norm{q_1\widetilde\Psi}
    =\Lambda_t\widetilde\alpha_N(t).
\end{align*}
The same argument gives
\begin{equation*}
    |T_{11}|
    \leq\Lambda_t\widetilde\alpha_N(t).
\end{equation*}

It remains to estimate $T_{01}$.  We use the Moore--Penrose inverse of $\widehat n$ on the orthogonal complement of $\Ker\widehat n$.  Since $q_1q_2$ vanishes on $\Ker\widehat n$, inserting $\widehat n^{1/2}\widehat n^{-1/2}$ is legitimate:
\begin{align}\label{eq:T01-split}
    |T_{01}|
    &\leq
    \norm{\widehat n^{1/2}D_t p_1p_2\widetilde\Psi}
    \norm{\widehat n^{-1/2}q_1q_2\widetilde\Psi}.
\end{align}
For the first factor,
\begin{align*}
    \norm{\widehat n^{1/2}D_t p_1p_2\widetilde\Psi}^2
    &=\frac1N\sum_{j=1}^N
    \norm{q_jD_t p_1p_2\widetilde\Psi}^2.
\end{align*}
If $j\geq3$, then $q_j$ commutes with $D_t,p_1,p_2$, and \eqref{eq:p1D2p1-projected-square-lemma} gives
\begin{equation*}
    \norm{q_jD_t p_1p_2\widetilde\Psi}^2
    \leq \Lambda_t^2\norm{q_j\widetilde\Psi}^2.
\end{equation*}
The two remaining indices are estimated directly by \eqref{eq:p1D2p1-projected-square-lemma}. Hence,
\begin{equation}\label{eq:first-factor-counting}
    \norm{\widehat n^{1/2}D_t p_1p_2\widetilde\Psi}^2
    \leq
    \Lambda_t^2\left(\widetilde\alpha_N(t)+\frac2N\right).
\end{equation}
For the second factor, using $\sum_{j=1}^Nq_j=N\widehat n$ and symmetry,
\begin{align}\label{eq:second-factor-counting}
    N(N-1)\norm{\widehat n^{-1/2}q_1q_2\widetilde\Psi}^2
    &\leq
    \sum_{i,j=1}^N
    \left\langle\widetilde\Psi,\widehat n^{-1}q_iq_j\widetilde\Psi\right\rangle\nn\\
    &=N^2\left\langle\widetilde\Psi,\widehat n\widetilde\Psi\right\rangle
    =N^2\widetilde\alpha_N(t).
\end{align}
Combining \eqref{eq:T01-split}, \eqref{eq:first-factor-counting}, and \eqref{eq:second-factor-counting}, and using $N\ge2$, we obtain
\begin{equation*}
    |T_{01}|
    \leq
    2\Lambda_t\left(\widetilde\alpha_N(t)+\frac1N\right).
\end{equation*}
Together with the bounds for $T_{10}$ and $T_{11}$, \eqref{eq:comm-to-pDq} gives
\begin{equation}\label{eq:differential-alpha-final}
    \left|\frac{\dd}{\dd t}\widetilde\alpha_N(t)\right|
    \leq
    8\Lambda_t\left(\widetilde\alpha_N(t)+\frac1N\right).
\end{equation}
Applying Gronwall's inequality to $\widetilde\alpha_N(t)+N^{-1}$ yields \eqref{eq: alpha gronwall bound main}.
\end{proof}

It remains to prove Lemma~\ref{lem:condensate-cancellation-projected-square}.

\begin{proof}[Proof of Lemma~\ref{lem:condensate-cancellation-projected-square}]
For notational simplicity, we omit the time dependence and write
\[
    \Phi=\Phi_t,
    \qquad
    \rho=\rho_{\Phi_t}=\rho_{\gamma_t},
    \qquad
    W=V\ast\rho,
    \qquad
    p_j=p_{j}^t,
    \qquad
    D=D_t.
\]
We use the identification
\[
    \widetilde\gH_N
    \cong
    L^2(\Omega^N;\gX^{\otimes N}).
\]

We first prove the cancellation
\[
    p_2Dp_2=0.
\]
It suffices to prove
\[
    p_2V(x_1-x_2)p_2=W(x_1)p_2.
\]
Let $\Psi,\Xi\in \widetilde\gH_N$ be arbitrary.  Since $p_2$ is the
rank-one projection onto $\Phi$ in the second lifted particle variable, we may
write
\[
    p_2\Psi=\Phi(x_2)\otimes \eta(x_1,x_3,\ldots,x_N),
\]
and
\[
    p_2\Xi=\Phi(x_2)\otimes \zeta(x_1,x_3,\ldots,x_N),
\]
for suitable
\[
    \eta,\zeta\in L^2(\Omega^{N-1};\gX^{\otimes(N-1)}).
\]
Then
\begin{align*}
&\inprod{p_2\Xi}{V(x_1-x_2)p_2\Psi} \\
&\quad =
\int_{\Omega^N}
V(x_1-x_2)
\norm{\Phi(x_2)}_{\gX}^2
\inprod{\zeta(x_1,x_3,\ldots,x_N)}
       {\eta(x_1,x_3,\ldots,x_N)}_{\gX^{\otimes(N-1)}}
\,\d x_1\cdots\d x_N .
\end{align*}
Since
\[
    \norm{\Phi(x_2)}_{\gX}^2=\rho(x_2),
\]
the integration in $x_2$ gives
\[
    \int_\Omega V(x_1-x_2)\rho(x_2)\,\d x_2
    =
    (V*\rho)(x_1)
    =
    W(x_1).
\]
Therefore,
\begin{align*}
&\inprod{p_2\Xi}{V(x_1-x_2)p_2\Psi} \\
&\quad =
\int_{\Omega^{N-1}}
W(x_1)
\inprod{\zeta(x_1,x_3,\ldots,x_N)}
       {\eta(x_1,x_3,\ldots,x_N)}_{\gX^{\otimes(N-1)}}
\,\d x_1\d x_3\cdots\d x_N  \\
&\quad =
\inprod{p_2\Xi}{W(x_1)p_2\Psi}.
\end{align*}
Hence,
\[
    p_2V(x_1-x_2)p_2=W(x_1)p_2,
\]
and consequently
\[
    p_2Dp_2
    =
    p_2\bigl(V(x_1-x_2)-W(x_1)\bigr)p_2
    =
    0.
\]

We now prove the projected square estimate.  Set
\[
    d(x,y):=V(x-y)-W(x).
\]
Thus, $D$ is multiplication by $d(x_1,x_2)$.  Let $\Psi\in\widetilde\gH_N$.

Define
\[
    \eta:=(\bra{\Phi}_1\otimes\1_{2,\ldots,N})\Psi
    \in
    \widetilde\gH_{N-1}.
\]
Then
\[
    p_1\Psi=\Phi(x_1)\otimes\eta(x_2,\ldots,x_N).
\]
Therefore,
\begin{align*}
    \norm{Dp_1\Psi}^2
    &=
    \int_{\Omega^N}
    \abs{d(x_1,x_2)}^2
    \norm{\Phi(x_1)}_{\gX}^2
    \norm{\eta(x_2,\ldots,x_N)}_{\gX^{\otimes(N-1)}}^2
    \,\d x_1\cdots\d x_N   \\
    &=
    \int_{\Omega^{N-1}}
    \left[
        \int_\Omega
        \rho(x_1)
        \abs{V(x_1-x_2)-W(x_1)}^2
        \,\d x_1
    \right]
    \norm{\eta(x_2,\ldots,x_N)}^2
    \,\d x_2\cdots\d x_N .
\end{align*}
By the definition of $\Lambda_t$, the expression in square brackets is bounded
by $\Lambda_t^2$ for almost every $x_2$.  Hence,
\[
    \norm{Dp_1\Psi}^2
    \le
    \Lambda_t^2
    \int_{\Omega^{N-1}}
    \norm{\eta(x_2,\ldots,x_N)}^2
    \,\d x_2\cdots\d x_N.
\]
Since
\[
    \norm{\eta}=\norm{p_1\Psi}\le \norm{\Psi},
\]
we obtain
\[
    \norm{Dp_1\Psi}^2
    \le
    \Lambda_t^2\norm{p_1\Psi}^2
    \le
    \Lambda_t^2\norm{\Psi}^2.
\]
Thus, $Dp_1$ extends to a bounded operator and
$
    \norm{Dp_1}_{\rm op}\le \Lambda_t.
$

Finally, since $D$ is self-adjoint as a multiplication operator,
\[
    p_1D^2p_1=(Dp_1)^*(Dp_1)
\]
in the sense of quadratic forms.  Consequently,
\[
    \norm{p_1D^2p_1}_{\rm op}
    =
    \norm{(Dp_1)^*(Dp_1)}_{\rm op}
    =
    \norm{Dp_1}_{\rm op}^2
    \le
    \Lambda_t^2.
\]
This proves \eqref{eq:p1D2p1-projected-square-lemma}.
\end{proof}


\end{document}